\documentclass[11pt]{amsart}
\usepackage{amscd}
\usepackage[arrow,matrix]{xy}
\usepackage{graphicx}
\usepackage{amsmath}
\usepackage{color}

\usepackage{amsmath, latexsym, amssymb}
\input xypic
\numberwithin{equation}{section}
\theoremstyle{plain}
\newtheorem{lemma}{Lemma}[section]
\newtheorem{proposition}[lemma]{Proposition}
\newtheorem{theorem}[lemma]{Theorem}
\newtheorem{corollary}[lemma]{Corollary}

\theoremstyle{definition}
\newtheorem{definition}[lemma]{Definition}
\newtheorem{remark}[lemma]{Remark}
\newtheorem{example}[lemma]{Example}

\begin{document}

\newcommand{\R}{{\mathbb R}}
\newcommand{\C}{{\mathbb C}}
\newcommand{\F}{{\mathbb F}}
\renewcommand{\O}{{\mathbb O}}
\newcommand{\Z}{{\mathbb Z}} 
\newcommand{\N}{{\mathbb N}}
\newcommand{\Q}{{\mathbb Q}}
\renewcommand{\H}{{\mathbb H}}
\newcommand{\Aa}{{\mathcal A}}
\newcommand{\Bb}{{\mathcal B}}
\newcommand{\Cc}{{\mathcal C}}    
\newcommand{\Dd}{{\mathcal D}}
\newcommand{\Ee}{{\mathcal E}}
\newcommand{\Ff}{{\mathcal F}}
\newcommand{\Gg}{{\mathcal G}}    
\newcommand{\Hh}{{\mathcal H}}
\newcommand{\Kk}{{\mathcal K}}
\newcommand{\Ii}{{\mathcal I}}
\newcommand{\Jj}{{\mathcal J}}
\newcommand{\Ll}{{\mathcal L}}    
\newcommand{\Mm}{{\mathcal M}}    
\newcommand{\Nn}{{\mathcal N}}
\newcommand{\Oo}{{\mathcal O}}
\newcommand{\Pp}{{\mathcal P}}
\newcommand{\Qq}{{\mathcal Q}}
\newcommand{\Rr}{{\mathcal R}}
\newcommand{\Ss}{{\mathcal S}}
\newcommand{\Tt}{{\mathcal T}}
\newcommand{\Uu}{{\mathcal U}}
\newcommand{\Vv}{{\mathcal V}}
\newcommand{\Ww}{{\mathcal W}}
\newcommand{\Xx}{{\mathcal X}}
\newcommand{\Yy}{{\mathcal Y}}
\newcommand{\Zz}{{\mathcal Z}}

\newcommand{\zt}{{\tilde z}}
\newcommand{\xt}{{\tilde x}}
\newcommand{\Ht}{\widetilde{H}}
\newcommand{\ut}{{\tilde u}}
\newcommand{\Mt}{{\widetilde M}}
\newcommand{\Llt}{{\widetilde{\mathcal L}}}
\newcommand{\yt}{{\tilde y}}
\newcommand{\vt}{{\tilde v}}
\newcommand{\Ppt}{{\widetilde{\mathcal P}}}
\newcommand{\bp }{{\bar \partial}} 

\newcommand{\ad}{{\rm ad}}
\newcommand{\Om}{{\Omega}}
\newcommand{\om}{{\omega}}
\newcommand{\eps}{{\varepsilon}}
\newcommand{\Di}{{\rm Diff}}
\newcommand{\Aff}{{\rm Aff}}

\renewcommand{\a}{{\mathfrak a}}
\renewcommand{\b}{{\mathfrak b}}
\newcommand{\e}{{\mathfrak e}}
\renewcommand{\k}{{\mathfrak k}}
\newcommand{\pg}{{\mathfrak p}}
\newcommand{\g}{{\mathfrak g}}
\newcommand{\gl}{{\mathfrak gl}}
\newcommand{\h}{{\mathfrak h}}
\renewcommand{\l}{{\mathfrak l}}
\newcommand{\n}{{\mathfrak n}}
\newcommand{\s}{{\mathfrak s}}
\renewcommand{\o}{{\mathfrak o}}
\newcommand{\so}{{\mathfrak so}}
\renewcommand{\u}{{\mathfrak u}}
\newcommand{\su}{{\mathfrak su}}
\newcommand{\ssl}{{\mathfrak sl}}
\newcommand{\ssp}{{\mathfrak sp}}
\renewcommand{\t}{{\mathfrak t }}
\newcommand{\Cinf}{C^{\infty}}
\newcommand{\la}{\langle}
\newcommand{\ra}{\rangle}
\newcommand{\half}{\scriptstyle\frac{1}{2}}
\newcommand{\p}{{\partial}}
\newcommand{\notsub}{\not\subset}
\newcommand{\iI}{{I}}               
\newcommand{\bI}{{\partial I}}      
\newcommand{\LRA}{\Longrightarrow}
\newcommand{\LLA}{\Longleftarrow}
\newcommand{\lra}{\longrightarrow}
\newcommand{\LLR}{\Longleftrightarrow}
\newcommand{\lla}{\longleftarrow}
\newcommand{\INTO}{\hookrightarrow}

\newcommand{\QED}{\hfill$\Box$\medskip}
\newcommand{\UuU}{\Upsilon _{\delta}(H_0) \times \Uu _{\delta} (J_0)}
\newcommand{\bm}{\boldmath}

\title[Circuit size of  partially homogeneous  polynomials]{Lower bounds for the circuit size   of  partially homogeneous polynomials}
\author{H\^ong V\^an L\^e} 
\thanks{H.V.L. is  supported by RVO: 67985840}
\address{Institute of Mathematics of ASCR, Zitna 25, 11567 Praha\\ email: hvle@math.cas.cz}

\medskip

\abstract{In  this paper
we  associate   to  each    multivariate  polynomial  $f$ that is homogeneous  relative to a subset of its variables  a series   of polynomial  families  $P_\lambda (f)$ of $m$-tuples of homogeneous polynomials of equal degree  such that   the circuit  size of any member in $P_\lambda (f)$ is bounded from above  by  the circuit  size of $f$.     This provides a method for obtaining lower  bounds  for
the circuit  size  of $f$  by proving  $(s,r)$-(weak) elusiveness  of the polynomial mapping associated with $P_\lambda (f)$.  We discuss  some algebraic methods  for proving the $(s,r)$-(weak) elusiveness.
We  also improve  estimates in the   normal  homogeneous-form  of an arithmetic circuit  obtained by Raz in \cite{Raz2009}  which   results  in better   lower bounds  for circuit size (Lemma \ref{lem:cor38}, Remark \ref{rem:cor38}).
 Our methods  yield  non-trivial  lower  bound  for  the circuit size of  several  classes of  multivariate   homogeneous  polynomials (Corollary \ref{cor:412}, Example \ref{ex:bi}). }
\endabstract
\maketitle

{\it  To my Teacher   Anatoly Timofeevich Fomenko}


\tableofcontents

\section{Preface}
I had a   fortune   to  study under the   guidance  of Anatoly Timofeevich  for almost all the time I spent  at the Lomonosov Moscow State University. I am greatly  indebted  
to Anatoly  Timofeevich  for his support and  encouragement,   for his  lessons in mathematics  and beyond. 
During my time  at the Lomonosov Moscow  University   Anatoly Timofeevich led, except  regular seminars on topology and differential geometry,
also a seminar on computer  geometry.  My present contribution  in computational complexity       reflects, in particular,  interests  I acquired     in Moscow  under  the influence of   Anatoly Timofeevich. I wish him  good health, happiness and  more success for the coming  years.   

\section{Introduction}
Let $\F$ be a field. Recall that the permanent $Per_n (\F)\in \F[x_{ij}|\, 1\le i, j\le n]$  is defined by
$$ Per_n ([x_{ij}]):= \sum _{\sigma \in \Sigma_n} \Pi_{i =1} ^n  x_{i\sigma(i)}.$$
Finding    non-trivial lower bounds  for the circuit size  or formula size  of the permanent $Per_n$ is a   challenging  problem  in    algebraic computational complexity  theory, especially in     the   $VP$ versus $VNP$   problem \cite{BCS1997}, \cite{Gathen1987}, \cite{Valiant1982}, \cite{SY2010}.
It has been pointed  out by  Mulmuley-Sohoni \cite{M-S2001} that  a   proof  of $VP\not = VNP$ which is  based on  a   generic property of  $poly(n)$-definable    polynomials  will likely  fall in the trap   of the ``natural  proof" \cite{RD1997}. 
Up to now,  there is no known tool  for obtaining  a non-trivial lower bound for the circuit size of the permanent.  The only known  non-trivial lower bound for the formula  size of the permanent is due to  Kalorkoti \cite{Kalorkoti}, which  says that  over any field, the  formula size  of  $Per_n(\F)$ is at least $\Om(n^3)$. (Kalorkoti  proved   the same lower bound for   the formula size of the determinant,   and   Pavel Hrube\v s told me that     Kalorkoti's   proof  works also for the   formula size of the  permanent.) Another  tool for obtaining a non-trivial  lower bound  for the formula  size of the permanent  exploits the Valiant theorem  on the relation between  the formula size and the determinantal complexity  of the permanent \cite{Valiant1979}, \cite{M-R2004}, \cite{M-S2001}.
The determinantal complexity  $c_{det}$, though   better understood than the formula size, is still   very complicated.  The best  lower bound   $c_{det}(Per_n)\ge (n^2/2)$ has been  obtained by  Mignon and Ressayer \cite{M-R2004}.  
To get the quadratic estimate, Mignon and Ressayer   compared  the   second fundamental form  of the  hyper-surface $\{ \det_m (x) = 0\}$ with that of $\{ Per_n(x) = 0\}$.  Mulmuley  and Sohoni suggested  to use  representation theory  to obtain   lower bounds  for $c_{det}(Per_n)$ \cite{M-S2001}.  We also like  to mention the recent paper  \cite{FLMS2013} on   reduction  to  circuits of depth 4. 

In \cite{Raz2009}  Raz introduced   new exciting ideas  to the  study  of lower bounds for the circuit size of   multivariate  polynomials. He   proposed a method  of elusive functions to construct polynomials   of large  circuit  size. Namely from an   $(s,r)$-elusive polynomial mapping  $f:\F^n \to \F^m$, for certain values of $(s,r,n,m)$,
he         obtained  a  multivariate  polynomial $\tilde f\in \F[x_1, \cdots, x_{3n}]$, whose  degree  linearly  depends on $r$,  such that the  circuit size $L(\tilde f)$ of $\tilde f$  is bounded  from below     by a function  of $r$ and $s$.  In  \cite{Le2010} we  developed further  Raz's ideas,  showing the effectiveness of his method
for  fields $\F = \R$ or $\F = \C$.

In this paper  we     develop  Raz's  ideas   in  a somewhat  different direction.  From  a  given     polynomial  $\tilde f$   that is  homogeneous  relative  to a subset of its variables (e.g.  the permanent, see 
Definition \ref{def:ph} below)  we  construct  a polynomial family  $P_\lambda (\tilde f)$ 
of $m$-tuples of  homogeneous   polynomials of  degree $r$  such that  the  $(s,2r-1)$-weak  elusiveness  of  the  polynomial mapping associated with $P_\lambda (\tilde f)$  would imply  a lower bound  for the circuit size of $\tilde f$ in terms of $s$  and $r$.  We  propose several  algebraic
methods  for  verifying   whether  a   homogeneous  polynomial mapping is  $(s,r)$-(weakly) elusive. We show  that our methods   yield  non-trivial  lower  bounds  for
a large class of  homogeneous polynomials. We  discuss some problems in commutative  algebra related with our method. 


The remainder of our paper is  organized as follows.  In Section \ref{sec:norm}  we   recall  basic  notions in the theory of  arithmetic circuits  that are needed  in our paper.
Then we  give a  slightly extended version of the Raz normal form theorem (Theorem \ref{thm:norm})      as well as an improved version of the Rax existence of a universal circuit-graph (Theorem \ref{prop:universal})  in the form that   is needed  in Proposition \ref{prop:eluh}, Corollary \ref{cor:his}  and Lemma  \ref{lem:cor38}. Lemma \ref{lem:cor38} is an improvement of a previous result  by Raz, see   Remark \ref{rem:cor38}. In section \ref{sec:2}  we  relate the notion of  (weakly)-elusive polynomial  mappings with  the circuit size of  a polynomial family  of $m$-tuples of  homogeneous polynomials of equal degree (Proposition \ref{prop:eluh}, Corollary \ref{cor:his}).  In section  \ref{sec:3} we   propose  several algebraic methods   for   proving the $(s,r)$-(weak) elusiveness of
a  polynomial mapping (Proposition \ref{prop:counting}, Remark \ref{rem:nonhom}, Corollaries \ref{cor:counting}, \ref{cor:keluh}, Examples \ref{ex:veronese}, \ref{ex:raz1}). We also  consider  related  problems in commutative  algebra (Problems 1,2,3).  In section \ref{sec:4}  we   associate to each    polynomial $\tilde f$ that is homogeneous
relative to a  subset of its variables  a series  of   polynomial  families  $P_\lambda (\tilde f)$ of  $m$-tuples  of homogeneous  polynomials of equal degree such that  the circuit size of any  member in $P_\lambda (\tilde f)$  is  bounded  from above by the circuit size  $L(\tilde f)$  of $\tilde f$  (Proposition \ref{prop:natu}).  We present  non-trivial   examples of our methods (Examples \ref{ex:raz}. \ref{ex:bi}). We also suggest a method  for obtaining non-trivial lower bounds for the  circuit size of the permanent (Lemmas \ref{lem:span}, \ref{lem:per}).

  
 \
  
{\bf Notations}. In our paper we     assume that $\F$ is an arbitrary field, if not   specified  otherwise.  
The space  of all (resp. homogeneous)  polynomials of degree  $r$ in $n$ variables over $\F$   will be denoted by
$Pol^r(\F^n)$ (resp. $Pol^r_{hom}(\F^n)$), and the space of all ordered $m$-tuples  of (resp. homogeneous) polynomials in $Pol^r(\F^n)$  (resp.  $Pol^r_{hom}(\F^n)$) will be denoted  by $(Pol^r(\F^n))^m$ (resp. $(Pol^r_{hom}(\F^n))^m$). We denote  by $Pol^r(\F^n, \F^m)$ (resp.  $Pol^r_{hom}(\F^n, \F^m)$) the space  of  polynomial mappings (resp. homogeneous polynomial mappings)  of degree $r$ from
$\F^n$  to $\F^m$. If $m = 1$ then   we abbreviate $Pol^r(\F^n, \F)$ as $Pol^r(\F ^n) ^*$.  Clearly, there is a natural linear map
$Pol^r(\F^n) \to Pol^r(\F^n)^*$, which is an isomorphism if $\F$  is a field of characteristic  $0$. We also note that there is a linear isomorphism 
$Pol^r(\F^n, \F^m) = (Pol^r(\F^n) ^ *)^m$.  For $\tilde f \in  Pol^r(\F^n) $ we denote by $\tilde  f^*$ the image of $\tilde f$ under the linear map 
$Pol^r(\F^n) \to Pol^r(\F^n)^*$. For $\lambda \in \F$  we also denote $\tilde f^*(\lambda)$  by $\tilde f(\lambda)$. 
 
 \section{Normal  form of arithmetic circuit and  their universal  circuit-graph}\label{sec:norm} 
 
 In this  section we recall some necessary definitions related  to arithmetic circuits.
We formulate  a  version of the  Raz theorem  on normal-homogeneous  circuit (Theorem \ref{thm:norm}) and a version of the   Raz theorem on the existence of  a universal circuit-graph (Theorem \ref{prop:universal}). These   results  are needed in  Proposition \ref{prop:eluh}, Corollary \ref{cor:his}  that relate the $(s,r)$-weak elusiveness  with a lower bound for the circuit size  of an arithmetic circuit.  This  results in a better  estimate on  the circuit  size  than  that of  Raz, see Lemma \ref{lem:cor38}  and Remark \ref{rem:cor38}.

\begin{definition} \label{def:ac} (cf. \cite[\S 1.1]{Raz2009}) {\it An arithmetic circuit} is a finite  directed acyclic graph whose  nodes are divided into four types:  {\it an input-gate} is a  node of in-degree 0  labelled with  an input variable; {\it a simple gate} is a node of in-degree 0 labelled with the field element 1; {\it  a sum-gate} is a node labelled with $+$ ; {\it a product-gate}  is a node labelled with $\times$; {\it an  output-gate}  is node of out-degree 0  giving the result of the computation. Every edge $(u,v)$ in the graph is labelled with a field element $\alpha$. It computes  the product of $\alpha$ with the polynomial computed by $u$. A product-gate (resp. a sum-gate) computes the product (resp. the sum) of  polynomials computed by the
 edges  that reach it.  We say that a polynomial  $g \in \F [x_1, \cdots, x_n]$ is computed by a circuit if it is computed by one of the circuit output-gates. If a  circuit has $m$ output-gates, then it computes an $m$-tuple of polynomials
$g^i \in \F [x_1, \cdots, x_n], \, i \in[1,m]$. 
The {\it  fanin} of a circuit is
defined to be the maximal in-degree of a node in the circuit, that is, the maximal number
of children that a node has.
\end{definition}

\begin{definition} \label{cgraph} (\cite[\S 2]{Raz2009}) {\it A circuit-graph  $G$} is  the underlying  graph $G_\Phi$ of an arithmetic circuit $\Phi$ together with the labels
of all nodes. This is  the entire  circuit, except for the labels of the edges. We call $G = G_\Phi$ {\it the circuit graph} of $\Phi$.  
 The {\it size}   of an arithmetic circuit $\Phi$  is defined  to be the number of edges in $\Phi$, and is denoted by $Size(\Phi)$. The {\it depth} of a circuit $\Phi$ is defined to be the length of the longest directed path in $\Phi$, and is denoted by $Depth(\Phi)$.
 The {\it circuit size}  $L(P)$ of an $m$-tuple  $P$ of  polynomials $g^1, \cdots, g^m \in \F[x_1, \cdots, x_n]$ is  the minimal size of an arithmetic circuit computing $P$.  
\end{definition}

\begin{definition}\label{def:syntax} For a circuit-graph $G$, we define {\it the syntactic-degree of a node in $G$} inductively as
follows \cite[\S 2]{Raz2009}. The syntactic-degree of a simple gate is 0,
and  the syntactic-degree of an input-gate is 1. The syntactic-degree of a sum-gate is the
maximum of the syntactic-degrees of its children. The syntactic-degree of a product-gate is
the sum of the syntactic-degrees of its children.
For an arithmetic circuit $\Phi$ and a node $v \in \Phi$, we define the syntactic-degree of $v$ to
be its syntactic-degree in the circuit-graph $G_\Phi$. The degree of a circuit is the maximal
syntactic-degree of a node in the circuit.
\end{definition}

\begin{definition}\label{def:normal} (\cite[Definitions 2.1, 2.2]{Raz2009}) 
A  circuit-graph  $G$ is called {\it  homogeneous}, iff  for every arithmetic circuit $\Phi$  such that   $G = G_\Phi$  and every gate $v$ in $\Phi$,  the polynomial computed by the gate $v$ is homogeneous.
Further, we  say that a homogeneous   graph  is {\it in normal form}, if it satisfies the following conditions.
\begin{enumerate}
\item There is no simple gate.
\item  All edges from the input-gates are to sum-gates.
\item All output-gates  are sum-gates.
\item The gates of $G$ are alternating. That is,  if $v$ is a product-gate  (resp.a sum-gate)  and $(u,v)$ is an edge, then $u$ is  a sum-gate (resp. a product-gate  or  an input-gate.)
\item The in-degree of  every product-gate is exactly 2.
\item The out-degree of every sum-gate is at most 1.
\end{enumerate}
We say that an arithmetic circuit  is  {\it in a normal-homogeneous  form}, if the circuit graph  $G_\Phi$ is in a normal-homogeneous form.
\end{definition}

Let $N(\Phi)$  denote  the number  of gates in $\Phi$.  

\begin{theorem}\label{thm:norm}(cf. \cite[Proposition 2.3]{Raz2009})  Let $\Phi$ be an  arithmetic circuit  of size  $s$ that  computes an $m$-tuple  $P$ of homogeneous  polynomials $g_1, \cdots, g_m \in Pol^r_{hom}(\F^n)$ where $r\ge 1$.
   Then there exists an arithmetic circuit $\Psi$ for the polynomials $g_1, \cdots, g_m$  such that $\Psi$ is in a normal homogeneous form  with  $N (\Psi) < 16s(r+1)^2 +5m +4n$. 
\end{theorem}

\begin{proof} Theorem \ref{thm:norm}  is  almost identical  with  \cite[Proposition 2.3]{Raz2009}   except that   Raz  assumed  that  $m = n$. 
 The  proof   presented here   uses  the Raz algorithm in  the proof of \cite[Proposition 2.3]{Raz2009}. The Raz algorithm transforms   an arithmetic  circuit $\Phi$ that computes $P$ into   an arithmetic circuit  $\Psi$   in  normal homogeneous  form  which also computes $P$ and, moreover, satisfies the condition of Theorem  \ref{thm:norm}.
 We shall write  the proof  of Theorem \ref{thm:norm} in detail, since  it will be needed  for the proof  of  Proposition  \ref{prop:universal} later. 
 
{\it Step 1}. If   a (sum- or product-) gate in  $\Phi$  has in-degree  1, then  we remove  its   and  connect  its  only child directly to  all its  parents.  The size 
 of the new  circuit   is less than the   size of the old circuit. Hence we can assume that $\Phi$ has {\it no  gate  of in-degree 1}.
(This property  is necessary for the next step and   needs not be preserved   under   later steps). 

{\it Step 2}. We transform $\Phi$ to $\Phi_1$, which satisfies  {\it the condition (5)}  of Definition \ref{def:normal},  by  replacing  any product-gate of in-degree larger than 2 with a tree of product-gates of in-degree 2  such that  each new born  product-gate has out-degree  one,  and by replacing any sum-gate of in-degree larger than 2 with a tree of sum-gates of in-degree 2  such that each new born sum-gate has out-degree one.  It is easy to  check that $Size(\Phi_1) \le 2s$. %

{\it Step 3}.  We transform $\Phi_1$ to $\Phi_2$    such that $G_{\Phi_2}$ also {\it  satisfies  the condition (5), and moreover, is homogeneous}. The nodes  of $\Phi_2$ are  obtained  by  splitting  each node $v\in \Phi_1$  into $(r+1)$ nodes $v_0, \cdots, v_r$, where the node $v_i$ computes the homogeneous part of degree $i$ of the polynomial computed by the node $v$. We ignore monomials  of degree larger than $r$. 
If the original node $v\in \Phi_1$  is    a  sum-gate,    we replace   the  sub-circuit in $\Phi_1$ connecting $v$ with its children $u ^1,\cdots, u^t$ by  the circuits  that compute $v_i = u_i^1 + \cdots + w_i^t$  for all $i \in [0, r]$.  
If $v\in \Phi_1$ is  a product-gate,  we replace the  sub-circuit in $\Phi_1$ connecting $v$ with its children $u^1, u^2$ by  the sub-circuits    that compute $v_i = \sum_{j = 0}^i u_j^1 \times u^2_{i-j}$  for all $ i \in [0, r]$. Clearly $\Phi_2$  also computes  $P$, moreover  $\Phi_2$ is homogeneous, satisfies  the condition (5)  in Definition \ref{def:normal}.
 By the  construction
  
\begin{equation}
Size(\Phi_2) \le r(r+1) Size (\Phi_1) \le 2s(r+1)^2.\label{eq:circ1}
\end{equation}
{\it Step 4}. We transform $\Phi_2$  to  a homogeneous  circuit $\Phi_3$ which computes $P$  and satisfies   {\it the conditions (1), (5)}  in Definition \ref{def:normal} by removing  every node of  syntactic-degree 0 as follows.  Let  $u\in \Phi_2$ be a  node of syntactic degree 0.
We  assume that  $u$ has out-degree  at least 1, otherwise  we  can remove $u$ without affecting the  functionality  of the circuit.  Let $v$ be  a parent of $u$.  If $v$ is a sum-gate, noting that $\Phi_2$ is homogeneous, $v$   computes a field element $\alpha_v$. Then we  replace  the sub-circuit computing $v$  from its  children  by  a simple gate  and    label the  corresponding   edge   by $\alpha_v$. If $v$ is a product-gate,  then $v$ has  the only  two children $u$ and $w$, 
so we replace the  sub-circuit   consisting  of  $v$ together with all edges connecting with $v$  by   edges with  appropriate   label  connecting  $w$    with the  parents of $v$.  Repeating this process    we get the  desired  circuit $\Phi_3$  with   no newly created   gate    and $Size(\Phi_3) \le Size(\Phi_2)$. 

{\it Step 5}. We transform $\Phi_3$  to a homogeneous  circuit $\Phi_4$  which computes $P$ and satisfies {\it the  conditions (1), (5) and (4)}. This is done as follows.  For any edge $(u,v)$ such that $ u, v$ are both  product-gates
we add a dummy sum-gate in between them. For any edge $(u,v)$  such that  $u, v$  are both  sum-gates  we connect all the children of $u$ directly to $v$.
Clearly $Size(\Phi_4)\le 2 \, Size(\Phi_3)$. 

{\it Step 6}. We transform $\Phi_4$  to  a homogeneous  circuit $\Phi_5$  which  computes $P$ and satisfies  {\it the conditions (1), (5), (4)  and  (3)} by connecting 
every product output-gate  to a new dummy sum-gate.  
Clearly
$$Size(\Phi_5) \le Size(\Phi_4) + m\le 2 \, Size(\Phi_2) +m.$$

{\it Step 7}. We transform $\Phi_5$ to  a homogeneous  circuit $\Phi_6$ which computes $P$  and satisfies {\it the conditions (1), (5), (4), (3) and (2)} by  adding a dummy sum-gate
in the middle of  any edge  from an input-gate to a product gate.  This  step  also transforms a formula $\Phi_5$ to  the formula $\Phi_6$.
Clearly 
$$Size(\Phi_6) \le 2 Size(\Phi_5) -m \le 4 \, Size (\Phi_2) +m .$$
{\it Step 8}. We transform $\Phi_6$ to   a homogeneous  circuit $\Phi_7$ which computes $P$  and satisfies all the conditions in Theorem \ref{thm:norm}  by duplicating  $q$-times any sum-gate of out-degree $q>1$. The resultant $\Phi_7$ may have  large  circuit size,  since
we do not have  a control  over the  number of edges  outgoing  from product-gate. Thus  we are  restrict  ourself with   the following  estimate
$$N(\Phi_7) \le 3 N(\Phi_6) \le  2 (Size (\Phi_6) + n +m)\le   8\, Size(\Phi_2) +5m +4n. $$
Taking into account (\ref{eq:circ1}), 
this completes  the proof of Theorem \ref{thm:norm}.
\end{proof}

\begin{theorem}\label{prop:universal} cf. \cite[Proposition 2.8]{Raz2009}  Assume that  a quadruple $(s,r,n, m)$ satisfies  $n,m\le s$, $1\le r$. Then  there is  a circuit-graph $G_{s,r,n,m}$, in a normal-homogeneous form 
that is universal for $n$-inputs and $m$-outputs  circuits  of size $s$ that computes  homogeneous  polynomials  of degree $r$, in the following sense.

Let $\F$ be a field. Assume that a $m$-tuple   $P : = (g_1, \cdots, g_m) \in (Pol^r_{hom}(\F^n))^m$ is of circuit size $s$. Then there exists an arithmetic  circuit  $\Psi$ that computes $P$  such that $G_\Psi = G_{s,r,n,m}$.

Furthermore,  the number of the edges  leading to the sum-gates in $G_{s,r,n,m}$  is less than $256\cdot s^2(r+2)^6$.
\end{theorem}

\begin{proof}  Theorem  \ref{prop:universal}  differs  from Proposition 2.8 in \cite{Raz2009} only in two instances.   Firstly,  Raz assumed that $m = n$. Secondly, 
we   have  an estimate on the number of the edges  leading to the sum-gates in $G_{s,r,n,m}$. This estimate, combined with  Remark \ref{rem:edges} below, yields  a better  lower bound  for the circuit size  of  partially      homogeneous  polynomials in considerations, see Remark \ref{rem:cor38}. 
 The idea  of the  proof of Theorem \ref{prop:universal}, due  to Raz \cite{Raz2009}, is to produce  a circuit-graph $G_{s,r,n,m}$ with   sufficient nodes and edges  so that the circuit-graph of  any normal-homogeneous circuit $\Phi$ computing
$P$   can be  embedded into $G_{s,r,n,m}$.   

Set $N := N(s,r,n,m) = 16\, s(r+1)^2 + 5m +4n$.

The  circuit-graph $G_{s,r,n,m}$ is constructed based on Theorem \ref{thm:norm}  as follows.  First  we describe how to  partition  the nodes of $G_{s,r,n,m}$ into $2r$ levels.
\begin{itemize}
\item The level-$1$  contains  $n$ input-gates, and the last level-$2r$  contains $m$ output-gates.
\item  For every $i \in \{ 2, . . . , r\}$, the level-$2i$ contains $N$ sum-gates of syntactic-degree $i$.
\item  For every $i \in \{ 2, . . . , r\}$, the level-$(2i-1)$ contains product-gates of syntactic-degree $i$.
\item  Every   product-gate  in level-$(2i-1)$  is  assigned  a  type $j\in [1, i-1]$.
\item  For   each pair $(i, j)$ such that $1\le j \le i-1\le r-1$  there are exactly  $N$  product-gates  of  syntactic degree $i$ and of type $j$.
\end{itemize}


Now  describe  the edges of the  circuit-graph $G_{s,r,n,m}$. First  we connect each sum-gate  in level-$(2i)$  with all  product-gates in level-$(2i-1)$.
Then we  connect   each  product-gate of type $j$  in level-$(2i-1)$ with    one sum-gate in level-$(2j)$  and  with one sum-gate
in level-$(2i-2j)$  inductively  using  an ordering  the set $\{(i, j)|\:1\le j \le i-1\le r-1\}$,   such that   the out-degree of every sum-gate is at most 1.

    Clearly  the constructed  circuit-graph $G_{s,r,n,m}$  is in  a normal-homogeneous form.


Let $P : = (g_1, \cdots, g_m) \in (Pol^r_{hom}(\F^n))^m$  have   a circuit  size $s$  and $\Psi $ an arithmetic  circuit  in normal homogeneous form that computes $P$ as  described in the proof in Theorem \ref{thm:norm}; in particular  $Size(P) \le 16\, s (r+1)^2 + 5m +4n$. 
We  will show how to embed  $G_\Psi$ into $G_{s,r,n,m}$. 

Since  $N(\Psi)\le N =  N (s,r,m,n)$,  we can embed  all the  product-gates of syntactic degree $i$   and of type $j$ of the circuit graph $G_{\Psi}$ into the  product-gates   of  type $j$  in level $(2i-1)$ of  $G_{s,r,n,m}$.   Since   the in-degree  of  each  product-gate  in $\Psi$ as well as in 
$G_{s,r,n,m}$   is two,   we   embed  all the   sum-gates  of $\Psi$ into  the sum-gate  of $G_{s,r,n,m}$ so that  the  edges   leading  to the  product-gates in $\Psi$   are  also   edges leading  to the  product-gates  in $G_{s,r,n,m}$.  Since   each  sum-gate   in level-$(2i)$ is connected with each product-gate  in  level-$(2i-1)$   the  edges leading to the  sum-gates in $\Psi$  can be  embedded  into  the  edges leading to the sum-gates in $G_{s,r,n,m}$.
 This completes the proof of   the first  assertion  of  Theorem  \ref{prop:universal}. 
 
 To prove  the last assertion of  Theorem \ref{prop:universal} we note that  for $i \in [1, r]$ there  are  at most $(r-1)N$   product-gates   on  level-$(2i-1)$ and  there  are exactly  $N$ sum-gates   on level-$(2i)$ of the  universal  circuit-graph  $G_{s,r,n,m}$. Hence   the total  number  of the edges leading to the sum-gates in
 $G_{s,r,n,m}$   is  at most
$r \cdot  (r-1)N\cdot N <  256\, s^2 (r+2)^ 6$. This completes  the proof   of Theorem \ref{prop:universal}.
\end{proof}

 



\begin{remark} \label{rem:edges}  (\cite[3.2]{Raz2009}) Assume that  $\Phi$ is a normal homogeneous arithmetic circuit  that computes  a $m$-tuple  $P \in (Pol^r_{hom}(\F^n)^m)$.  Then  there is an arithmetic circuit $\Psi$ of the same circuit-graph  as $\Phi$  that computes $P$  such that    the label of  any edge leading to a product-gate in $\Psi$  is 1.
\end{remark}

\section{$(s,r)$-weakly elusive  polynomial mappings}\label{sec:2}

In this section we introduce the notion of an $(s,r)$-weakly elusive  polynomial mapping (Definition \ref{def:eluh}), which is slightly weaker than  the notion of  an $(s,r)$-elusive   polynomial mapping introduced by
Raz (Example \ref{ex:1}), see  also  Remark \ref{rem:eva} in Section 3  for motivation. Then we  show  how this notion  is useful for obtaining  lower bounds for  the circuit size  of elements in a polynomial family of $m$-tuples  of homogeneous  polynomials (Proposition \ref{prop:eluh}, Corollary \ref{cor:his}).  The key geometric  structures here are  polynomial families
of $m$-tuples of homogeneous  polynomials of equal degree (Definition \ref{def:polf}).  


\begin{definition}\label{def:eluh} (cf. \cite[Definition 1.1]{Raz2009})  A polynomial  mapping $f: \F^n \to \F^m$  is called  {\it $(s,r)$-weakly  elusive},  if its image
does not belong to the image  of any homogeneous  polynomial mapping $\Gamma: \F^s \to \F^m$ of degree   $r$.
\end{definition}

This definition differs  from   the Raz definition \cite[Definition 1.1]{Raz2009}  only in   the requirement  that $\Gamma$  must be homogeneous.  This is a minor  difference, as we  will see in the example below, but  it will be technical  simpler  in   some situations. 

\begin{example}\label{ex:1}  \begin{enumerate}
\item Any $(s,r)$-elusive polynomial  mapping is $(s,r)$-weakly elusive.
\item The curve $(1, x, \cdots, x^m) \in \R^{m+1}$ is $(m, 1)$-weakly elusive, since its  image does not belong to any hyper-surface  through the origin  of $\R^{m+1}$. On the other hand,   this curve is not $(m,1)$-elusive, since  it lies on the affine  hyper-surface $x_1 = 1$  in $\R^{m+1}$.
\item  If $f : \F^n \to \F^m$ is $(s+1, r)$-weakly elusive, then  $f$  is $(s, r)$-elusive.
\end{enumerate}
\end{example}


The notion of $(s,r)$-weakly elusive  polynomial mappings  aims to verify, whether elements in a  polynomial family of  $m$-tuples of homogeneous polynomials of
equal degree  have uniformly bounded circuit size.

Given a set $S$ of  variables $x_1, \cdots, x_l$  we denote by  $Pol ^r(\F\la S\ra)$ (resp. $Pol^r_{hom}(\F\la S \ra)$)  the space of
polynomials (resp  homogeneous polynomials) of degree $r$ in variables $x_1, \cdots, x_s$ over $\F$.


\begin{definition}\label{def:polf} A family $P_\lambda \in (Pol ^r_{hom} (\F^n))^m,\, \lambda \in  \F^k$,  will be called {\it  a polynomial family  of $m$-tuples of  homogeneous
polynomials of equal degree}, if  there   exists  a polynomial mapping $f: \F^k \to  \F^N = (Pol^r_{hom}(\F^n))^m$,  such that $P_\lambda =  f(\lambda)$ for all $\lambda \in \F^k$.  The polynomial mapping $f$ will be called {\it associated  with the  family $P_\lambda$}.
\end{definition}

\begin{proposition}\label{prop:eluh} Let  $Z = \{ z_1, \cdots, z_n \}$  be a set of variables  and  $1\le n,m\le s$. 
Assume that  $P_\lambda \in  (Pol ^r _{hom}(\F(\la Z \ra) ) ^m, \, \lambda \in \F^k,$    is  a polynomial family  of  $m$-tuples of homogeneous polynomials in $n$ variables of degree $r$ such that  for each $\lambda \in \F^k$ the circuit size of  $P_\lambda$ is  at most $L$.
Then  the associated   polynomial  mapping $f$  is not  $(s, 2r-1)$-weakly  elusive  for any $s \ge s_0:=  256\cdot L^2 \cdot ( r+2)^6$.
\end{proposition}

\begin{proof}  Let $s$ be  an integer as in Proposition \ref{prop:eluh}. To prove  Proposition \ref{prop:eluh} it suffices to show the existence of   a homogeneous polynomial mapping of degree $(2r-1)$
$$\Gamma_G : \F^s \to (Pol^r_{hom}(\F\la Z\ra))^m$$
such that      
\begin{equation}
f (\F^k)\subset \Gamma_G (\F^s).\label{eq:eluh1}
\end{equation}

We shall construct a homogeneous  polynomial  mapping $\Gamma_G$   satisfying (\ref{eq:eluh1}) with help of Proposition \ref{prop:universal}. 
 By Theorem \ref{prop:universal}, the universal  circuit-graph  $G_{L,r,n,m}$   has at most  $s_0$  edges  leading  to  the sum-gates. 
We label  these edges with $y_1, \cdots, y_{\bar s}$, where $\bar s \le s_0$. We label the  other edges  of $G_{L,r,n,m}$  with  the field element $1$, see Remark \ref{rem:edges}.
Now  we define $\Gamma_G$  to be   the  polynomial  mapping in the variables $y_1, \cdots, y_s$   such that 
$$\Gamma_G(\alpha_1, \cdots, \alpha_s) = (g_1, \cdots, g_m)$$
where $(g_1, \cdots, g_m)$ are  the $m$ output-gates  of the  circuit $\Phi_{G_{L,r,n,m}}$ obtained  from $G_{L,r,n,m}$ by replacing   the label $y_i$ with the   field element $\alpha_i \in \F$  for  all $i \in [1, \bar s]$. In particular, $\Gamma_G$ depends   only on 
$\bar s$ variables.

\begin{lemma}\label{lem:gammag} (cf. \cite[Proposition 3.2]{Raz2009})  $\Gamma _G$ is a homogeneous polynomial  mapping of degree  $2r-1$.
\end{lemma}
\begin{proof}  We apply the  argument in the Raz proof  of \cite[Proposition 3.2]{Raz2009}, which is  a  special case  of Lemma \ref{lem:gammag} with $m = n$. For a node $v \in G_{L,r,n,m}$ denote  the polynomial $g_v\in \F [z_1, \cdots, z_n, y_1,\cdots, y_s]$ that is computed  by the node $v$ and is regarded as polynomial in $[z_1, \cdots, z_n]$ with coefficients in  $\F[y_1, \cdots, y_s]$. 
 If $v$ is a product-gate, with children $v_1 , v_2$ (that are sum-
gates), then, by induction, the coefficients in the polynomials $g_{v_1} , g_{v_2}$ are homogeneous
polynomials of degree $2r_{v_1} -1$, $2r_{v_2} -1$, respectively, (in the labels $y_1, \cdots ,y_s$). By Remark \ref{rem:edges}
$(v_1 , v)$ and $(v_2 , v)$ are labelled by 1, the coefficients in the polynomial $g_v$ are homogeneous
polynomials of degree $2r_{v_1} - 1 + 2r_{v_2} - 1 = 2r_v - 2$ (in the labels $y_1, \cdots ,y_s$).
If $v$ is a sum-gate, then, by induction, the coefficients in the polynomial $g_u$, for every
child u of v, are homogeneous polynomials of degree $2r_u - 2 = 2r_v - 2$ (in the labels
$y_1, \cdots , y_s$). Since the edge $(u, v)$ is labelled by an element of $\{ y_1 , \cdots , y_s \}$, the coefficients in
the polynomial $g_v$ are homogeneous polynomials of degree $2r_v - 1$ (in the labels $y_1 , \cdots, y_s$).
\end{proof}
 
By the assumption of Proposition \ref{prop:eluh}, for any $\lambda \in \F^k$, the  circuit size of $f(\lambda)$ is at most $L$. Taking into account Theorem \ref{prop:universal},     there exists $\alpha \in \F^s$ such that
$f (\lambda) = \Gamma_G(\alpha)$.  This proves (\ref{eq:eluh1}) and completes the proof  of Proposition \ref{prop:eluh}.
\end{proof}

\begin{corollary}\label{cor:his} 
 Let  $P_\lambda \in (Pol^r _{hom} (\F^n))^m$, $\lambda \in \F^k,$  be a  polynomial family  of $m$-tuples of homogeneous polynomials of degree $r$
   and $f: \F^k \to (Pol ^r_{hom} (\F^n))^m$  its  associated polynomial mapping.
Assume that $f$ is  $(s, 2r-1)$-weakly elusive.  Then
$P_\lambda$  has a member  with circuit size  greater than  or equal $\frac{\sqrt s}{16 (r+2) ^3}$.  
\end{corollary}

 Proposition \ref{prop:eluh}    crystallizes some arguments  in Raz's  proof of \cite[Proposition 3.7]{Raz2009}.
In Proposition \ref{prop:natu}  below we shall see that  for each multivariate  partially homogeneous polynomial $\tilde f$  there are many  polynomial  families 
$P_\lambda (\tilde f)$ of $m$-tuples
of homogeneous  polynomials of equal degree associated
with $\tilde f$ such that the   circuit size  of any member  in the family $P_\lambda (\tilde f)$ is  bounded   by the circuit size of $\tilde f$. Then we
can apply  Proposition \ref{prop:eluh}, or its  equivalent version Corollary \ref{cor:his},   for  estimating   from below the
circuit size  $L(\tilde f)$.

\section{How to prove  the $(s,r)$-weak elusiveness}\label{sec:3}
In this section we assume that  $\F$ is a field of characteristic 0, or the size of $\F$ is sufficiently large, so that $Pol^k_{hom} (\F^n, \F^m) = (Pol ^k_{hom}(\F^n))^m$. 

Given  $f\in Pol^k(\F^n, \F^m)$  and  two numbers $s,r$,  it is generally hard  
to know whether $f$  is    $(s,r)$-weakly elusive  or $(s,r)$-elusive. In this section  we propose  some algebraic  methods to establish   the $(s,r)$-(weak)  elusiveness of $f$ (Proposition \ref{prop:counting},  Remark \ref{rem:nonhom}, Corollaries \ref{cor:counting}, \ref{cor:keluh}, Examples \ref{ex:veronese}, \ref{ex:raz1}).  Under ``algebraic  methods" (resp. 
``algebraic  characteristics")  we  mean   operations    on $f$ (resp. properties like the dimension of vector spaces  associated  with $f$), which are  related with  techniques  developed  in commutative algebra.  We show that, for  appropriate parameters $(s,r,n,m, p)$, the subset of $(s,r)$-(weakly) elusive  homogeneous polynomial mappings is everywhere dense with respect to the Zariski topology in the  space $Pol^p_{hom}(\F^n, \F^m)$ (Theorem \ref{thm:exist}).   We also pose some problems in  commutative algebra (Problems 1, 2, 3)   whose solutions  would advance the proposed methods.

\subsection{Hilbert functions and $(s,r)$-weak elusiveness} 
Let  $  \F P^{m-1}$  denote the  projective space  of   dimension $(m-1)$  over   field $\F$, that is $\F  P^{m-1} = (\F^m \setminus \{ 0 \}) /(\F \setminus \{ 0\})$.
Elements  of $\F  P^{m-1}$  are denoted by $[x_1, \cdots, x_m]$, where $x_i \in \F$. 
For $\Gamma = (\Gamma ^1, \cdots, \Gamma^m)  \in  (Pol^r_{hom} (\F^s))^m   =  Pol ^r _{hom}(\F^s, \F^m)$    and $g \in  Pol ^*_{hom}(\F ^m)$  let
$$\Gamma ^*(g) : = g(\Gamma ^1, \cdots, \Gamma ^m) \in  Pol^*_{hom} (\F ^s),$$
$$\Gamma _{pr}:= \{ [\Gamma^1 (\lambda), \cdots, \Gamma ^m (\lambda) ]\in  \F P^{m-1}) |\, \lambda \in \F^s \}.$$
 (Thus $\Gamma_{pr}$ is the  projective variety in $\F P^{m-1}$ that is associated  with  $\Gamma$.)
 
For a  given  quadruple $(s,r,m,d)$  with $s\le m-1$ we set
$$l_{hom}(s,r,m,d) : = \max  \{ \dim  \Gamma ^* (Pol ^ d _{hom} (\F^m))|\, \Gamma \in  (Pol^r_{hom} (\F^s))^m\}.$$

Let\\
- $I^d_{hom} (\Gamma(\F^s))$  denote the space  consisting   of all homogeneous  polynomials of degree    $d$ in the ideal $I(\Gamma(\F^s))$,\\
- $A^d_{hom} (\Gamma)$ - the   quotient space $Pol^d_{hom} (\F^m)/ I^d_{hom} (\Gamma(\F^n))$.

Since $I^d_{hom}(\Gamma(\F^n)) = \ker \Gamma ^ * \cap Pol^d_{hom} (\F^m)$, we have
\begin{equation}
\dim \Gamma^*(Pol^d_{hom} (\F^m)) =  \dim A^d_{hom}(\Gamma).\label{eq:adg}
\end{equation}
Note that $\dim A^d_{hom}(\Gamma)$
is   equal  to the  value   $h_{\Gamma}(d)$ of {\it the  Hilbert  function  of the variety $\Gamma_{pr}$ at $d$}, 
 see for instance  \cite[Lecture 13]{Harris1993}.   
Hence it  follows from (\ref{eq:adg})
\begin{equation}
l_{hom}(s, r, m, d) = \max \{ h_{\Gamma} (d) |\, \Gamma \in Pol ^r _{hom}(\F^s, \F^m)\}.\label{eq:hil1}
\end{equation}
Furthermore, we denote by $Pol ^{rd} _{hom} (\F^s)$ the linear space   of   all homogeneous polynomials $g$ of   degree    $rd$ in $(x_1, \cdots, x_s)$.
Then  $\Gamma ^* (Pol ^ d _{hom} (\F^m)) \subset Pol ^{rd} _{hom} (\F^s)$ for all $\Gamma \in  Pol^r_{hom} (\F^s, \F^m)$. Hence we   obtain
\begin{equation}
l_{hom}(s, r, m, d)\le  \dim  Pol ^{rd} _{hom} (\F^s) = \binom{rd+s -1}{rd}.\label{eq:hil2}
\end{equation}

{\bf Problem  1.}  Find      upper bounds for  $l_{hom}(r,s,d, m)$ that are better than  (\ref{eq:hil2}).
Equivalently, we need to find  better upper  bounds  for $h_\Gamma(d)$ for $\Gamma \in Pol ^r _{hom}(\F^s, \F^m)$.

Now assume that  $f$  is a homogeneous polynomial  mapping    
from $\F ^n$  to $\F^m$, where $m\ge n+1 \ge  2$.

\begin{proposition}\label{prop:counting} Let $f$  be a homogeneous polynomial  mapping    
from $\F ^n$  to $\F^m$, where $m\ge n+1 \ge  2$. If there  exist   $s,r, d \ge 1$  
such that $h_f (d) \ge  l_{hom}(s,r,m,d) +1$,
then $f$ is $(s, r)$-weakly   elusive.  
\end{proposition}

\begin{proof} Assume that $f$ satisfies the  condition  of Proposition  \ref{prop:counting}. We will show that   for any homogeneous  mapping
$\Gamma: \F^s \to \F^m$ of degree $r$ the image of $f$  does not lie  on the image of $\Gamma$. Assume the opposite, i.e. there exists  a  homogeneous
 mapping $\Gamma:\F^s \to \F^m$ of degree  $r$ such that $f(\F^n) \subset \Gamma(\F^s)$. Then
\begin{equation}
I_{hom}(\Gamma (\F ^s)) \subset  I_{hom}(f(\F^n)).\label{eq:ideal1}
\end{equation}
 Let $I ^{\perp, d}_{hom} (f(\F^n))$ be   a  complement  of the  subspace $I^d_{hom}(f(\F^n))$ in $Pol^d_{hom} (\F^m)$. 
  Since $\Gamma$  is  homogeneous of degree
$r$, using (\ref{eq:hil1}), we have
\begin{equation}
\dim \Gamma ^* (I^{\perp, d}_{hom}(f(\F^n))) \le \dim \Gamma^ *(Pol^d_{hom} (\F ^m))\le l_{hom}(s,r,m,d).\label{eq:compare1}
\end{equation}
Taking into account $\ker \Gamma^ * \cap Pol_{hom}^*(\F^ m) = I_{hom}(\Gamma(\F^s))$, (\ref{eq:ideal1})  implies that
\begin{equation}
\dim \Gamma ^* (I^{\perp, d}_{hom}(f(\F ^n)))  = \dim  A ^d_{hom}(f)= h_f(d).\label{eq:compare2}
\end{equation}
Clearly (\ref{eq:compare1})  and (\ref{eq:compare2}) contradict the assumption of our Proposition.  This proves
 that $f$  is $(s,r)$-weakly elusive.
\end{proof}


\begin{remark}\label{rem:nonhom}  The above arguments   also  apply  to  the study of   $(s,r)$-elusive functions. We  set
$$l(s,r,m,d) : = \max  \{ \dim  \Gamma ^* (Pol ^ d _{hom} (\F^m))|\, \Gamma \in  (Pol^r (\F^s))^m\}.$$
The   argument   in the proof  of Proposition  \ref{prop:counting}  implies the following  assertion. Assume that $f$ is a homogeneous
polynomial mapping  from $\F^n$ to $\F^m$. If   for some $s,r, d \ge 1$ we have  $h_f (d) \ge  l (s,r, m, d) +1$  then $f$  is $(s,r)$-elusive.
\end{remark}

We obtain immediately from Proposition \ref{prop:counting}  and   Remark \ref{rem:nonhom} the following

\begin{corollary}\label{cor:counting} Assume that  $f$ is a homogeneous polynomial mapping  from $\F^n$ to $\F^m$.

1. If for some $d,s,r \ge 1$  
we have 
$ h_f(d) \ge \binom{rd+s -1}{rd} $,
then $f$ is $(s, r)$-weakly   elusive. 

2. If for some $d,s,r\ge 1$  
we have 
$ h_f(d) \ge \binom{rd+s}{rd} $,
then $f$ is $(s, r)$-elusive. 
\end{corollary}

Now we redenote $\binom{n-1+k}{k}$  as $b(n-1+k,k)$.

\begin{example}\label{ex:veronese}   Let us consider  the  Veronese  mapping  $\nu_k: \C^n \to \C^{b(n-1 +k)(k)}= Pol^k_{hom}(\C^n)$  of degree $k$:
$$\nu_k(x_1, \cdots, x_n) := (x_1 ^k, x_1^{k-1} x_2, x^{k-2}_1 x_2^2, \cdots, x_n^k).$$
It is known that  $A^d_{hom} (\nu _k)$  is  equal to $Pol^{dk}_{hom}(\C^n)$, see e.g. \cite[Example 13.4]{Harris1993}.
 By Corollary \ref{cor:counting}, $\nu_k$ is $(s, r)$-weakly elusive, 
 if  for some $d$ we have
\begin{equation}
\binom{dk +n -1}{dk} \ge \binom{rd+s-1}{rd} +1 .\label{eq:vero1}
\end{equation}
\end{example}

To  apply  Corollary  \ref{cor:counting} to solving the question  whether  a homogeneous polynomial mapping  $f$ from $\F^n$ to $\F^m$ is
 (weakly)  elusive,  in the first step, we   search
for an intermediate  lower  bound  for the  value $h_f (d)$  of the Hilbert function $h_f$, where $d$  is  some  appropriate integer.
The following Lemma suggests a  way to find   such a lower bound. 

\begin{lemma}\label{lem:lhilbert}  Let  $f$ be  a homogeneous polynomial mapping  from $\F^n$ to $\F^m$.
Assume  that  there exists  a subspace $\Ll \subset  Pol^d_{hom} (\F^m)$
such that $\ker  f^*  \cap \Ll = 0$.   Then $h_f (d) \ge  \dim \Ll$.
\end{lemma}
\begin{proof} As we  have observed above,  
$$h_f(d) = \dim  f^*(Pol^d_{hom} (\F^m)) \ge  \dim   f^* (\Ll) = \dim \Ll.$$
 The last   equality holds
since $\ker  f^*  \cap \Ll = 0$.  This   proves Lemma \ref{lem:lhilbert}.
\end{proof}

\begin{example}\label{ex:raz1}  As  an application of Corollary \ref{cor:counting}  and  Lemma \ref{lem:lhilbert}  we shall explain  Raz's  proof   of  $(s,d)$-elusiveness of functions  constructed  in \cite[Lemma 4.1]{Raz2009}.
For an integer $k$, denote by $[k]$ the set $\{ 1, \cdots , k\}$. Let $m = n^2$. We
identify the set $[m]$ with $[n] \times [n]$ by the lexicographic order. Let $1 \le d\le  (\log_2 n)/100$ be
an integer. Let $d' = 5d$. Let $X= \{ x_{i,j }\}_{i\in [d' ],j\in [n]}$ be a set of $n \cdot d'$ input variables. For every
$(a, b) \in [n] \times [n] = [m]$, define a polynomial
$$ f_{(a,b)} (x_{1,1} , \cdots , x_{d' ,n} ) =\Pi_{i\in [d']} x_{i,a+i\cdot b}$$
where the sum $a + i \cdot b$ is taken
modulo $n$. Let 
$$f = (f_{(1,1)} , f_{(1,2)} , \cdots  , f_{(n,n)} ).$$
Raz  proved  that  the polynomial mapping $f : \F^{n\cdot d'} \to  \F^m$ is $(s, d)$-elusive, where $s = \lfloor n^{1+1/(2d)}\rfloor $.
In his proof  Raz introduced the notion of
{\it a retrievable monomial in $ Pol ^r_{hom}(\F^m)$}, or equivalently, {\it a  retrievable  subset $Q \subset [n]\times [n]$ of size $r$}.
We define  a map $R:   2 ^{ [n] \times [n]} \to  \F[x_{1, 1}, \cdots  x_{d', n}]$ as follows
$$R(Q) =  f_Q : = \Pi_{(a, b) \in Q } f_{a,b} =  \Pi_{(a, b) \in Q }\Pi_{i \in [d']} x_{i, a + i b}.$$
Let $m_Q$  denote the monomial  $\Pi_{(a,b) \in Q} x_{(a,b)} \in Pol ^r_{hom}(\F^m)$. Then   $R(Q) =  f^*(m_Q)$.
A subset $Q\subset  [n]\times [n]$ is called {\it retrievable},  if $R^{-1} ( R (Q)) = Q$, or  equivalently $(f^*) ^{-1} ( f^*(m_Q)) = m_Q$.
  Raz proved the following 
  
 \
 
{\bf Claim  R} \cite[Claim 4.2]{Raz2009} The  set of retrievable   monomials $m_Q$ of degree $ r = \lfloor n ^{1- 1/(2d)} \rfloor$  is at least 
a half of the  set of  all   monomials  in $ Pol ^r_{hom}(\F^m)$.

\

Now  let $\Ll \subset Pol^r_{hom} (\F^m)$   be  generated  by  retrievable   monomials $Q$ of degree $ r = \lfloor n ^{1- 1/(2d)} \rfloor$.
It is not hard  to see  that  $\ker  f^*  \cap \Ll = 0$ \cite[Claim 4.4]{Raz2009}.  Consequently, Lemma  \ref{lem:lhilbert} and Claim R yield
\begin{equation}
h_f (r) \ge \dim \Ll \ge  {1\over 2} \binom{m}{r}.\label{eq:hfr}
\end{equation}
Furthermore, Raz  get the following estimates
\begin{equation}
{1\over 2} \binom{m}{r} \ge  s ^r  > \binom{rd +s }{rd}.\label{eq:rds}
\end{equation}
Using (\ref{eq:hfr}), (\ref{eq:rds})  and  Corollary  \ref{cor:counting}.2
we obtain  the $(s,d)$-elusiveness  of $f$. 
\end{example}

Thus, to   apply the Hilbert  function  method to the study of  (weak) elusiveness of  homogeneous polynomial mappings, we need  to  investigate the following.

\

{\bf  Problem 2.} For a given $f \in Pol^k_{hom}(\F^n, \F^m)$  find a lower bound for  $h_f(d)$.

\

Problems  1,2  are   related   to  the problems  of   searching for   lower bounds and upper bounds  of Hilbert functions. 
We refer the reader to \cite{Sombra1997}  for  an overview of  lower bounds and upper bounds  of Hilbert functions.

 
 

\subsection{ $(s,r)$-weakly elusive  subsets  and $(s,r)$-weakly elusive polynomial mappings}  In this subsection, adapting the  methods  of elusive subsets  developed in  \cite{Le2010},  we  reduce the  problem of verifying  whether a polynomial mapping $f : \F^n \to \F^m$ is $(s, r)$-weakly  elusive, to
 verifying  whether  a subset $A$ in the image  of $f(\F^n)$ is  $(s,r)$-weakly elusive.

\begin{definition}\label{def:eluhs}  A  subset $A \subset  \F^m$  will be called {\it $(s,r)$-weakly  elusive}, if $A$ does not lie  on the image of any homogeneous  
polynomial mapping $\Gamma: \F^s \to \F^m$ of degree  $r$.
\end{definition}

In order to prove that  $f$  is $(s,r)$-weakly elusive,  it suffices to  show the existence of  a  $k$-tuple  of points in the image of $f(\F ^n)$, which is
$(s,r)$-weakly elusive, i.e.   it does not lie on the image   of any  homogeneous  polynomial mapping $\Gamma: \F^s \to \F^m$ of degree $r$.
We regard   a $k$-tuple  $S_k = (b_1, \cdots, b_k)$, $b_i\in \F^m$, as  an element in $(\F^m )^k  = \F^{mk}$.

Recall that  we identify $Pol^r_{hom}(\F^s, \F^m)$ with $(Pol^r_{hom} (\F^s)^*)^m $.

\begin{proposition}\label{prop:keluh} (cf. \cite[Lemma 2.4]{Le2010})   A tuple   $S_k\in (\F^m)^k$ of $k$ points in $\F^m$ is $(s,r)$-weakly  elusive, if and only if  $S_k$  does not belong  to the image of the evaluation map
$$Ev ^r_{s, m, k}:  (Pol^r_{hom} (\F^s)^*)^m  \times (\F^s) ^k \to \F^{mk} ,$$
\begin{equation}
[( \tilde f_1 ^*, \cdots,  \tilde f_m^*) , (a_1, \cdots, a_k)] \mapsto ( \tilde f_1^*(a_1), \cdots,  \tilde f_m^* (a_k)).\label{eq:ev1}
\end{equation}
\end{proposition}

Proposition \ref{prop:keluh} is proved in the same way as \cite[Lemma 2.4]{Le2010}, so we omit its proof. (Note that we use here a   notation  for the evaluation mapping
$Ev^r_{s,m,k}$  which seems  better  than  the  notation $Ev ^k_{r,s,m}$ in \cite{Le2010}.)

\begin{remark}\label{rem:eva} 
 Our introduction of the notion  of weakly   elusive  functions is  motivated by the  fact that   the   evaluation  map $Ev ^r_{s, m, k}$ in Proposition \ref{prop:keluh}  is bi-homogeneous  with respect to the variables $(\tilde f_i ^*)$ and the variables $(a_i)$. So it is easier  to handle with  the new  evaluation map than  with the evaluation mapping associated with elusive functions. 
\end{remark}

\begin{corollary}\label{cor:keluh} (cf. \cite[Corollary 2.5]{Le2010}) A polynomial  mapping $f : \F^ n \to \F^m$  contains an $(s,r)$-weakly elusive
$k$-tuple, if and only if  the  subset
$$\hat f^k : =f(\F^ n) \underbrace{\times \cdots \times} _{k\text{ times}}  f(\F^n) \subset \F ^ {mk}$$
does not belong  to the image of the evaluation mapping $ Ev^ k_{s,r,m}$.
\end{corollary}

Clearly, the subset $\hat f^k$ does  not belong  to the image  of the  polynomial map $ Ev^r_{s,m,k}$, if   the Zariski closure
$\overline{\hat f^k}$ of   $\hat f^k$ does not belong to the Zariski closure $\overline{ Ev^ r_{s,m,k}}$ of the image of $ Ev^r_{s,m,k}$. Thus we   pose the  following  problem, whose  solution  is an important step  in proving  that a  polynomial mapping $f$ is  $(s,r)$-weakly  elusive.

{\bf Problem 3.}  Find  elements   of the  ideal  of $\overline{ Ev^ r_{s,m,k}}\subset (\F^m)^k$, i.e. elements
in the kernel of the ring homomorphism: $(Ev^r_{s,m,k})^*: \F[x_1, \cdots, x_{mk}] \to \F[y_1,\cdots,  y_N]$, $N = m\dim (Pol^r_{hom}(\F^s)^*) +ks$.

\

Once  we find  a  ``witness" $W$ in   $\ker   (Ev^r_{s,m,k})^*$,  we   could check   if $(f \times _{k\, times}  f)^*(W)  = 0$. If  not, then
the polynomial mapping $f$ is $(s,r)$-weakly elusive.

Problem  3 seems  very hard.  At the first step we should  study  property  of the   ideal  of $\overline{ Ev^ r_{s,m,k}}\subset (\F^m)^k$, which  could be sufficient for      proving  the weak elusiveness of   some concrete   polynomial mappings $f$, using Corollary  \ref{cor:keluh}.
Let us describe the  ideal  of $\overline{ Ev^ r_{s,m,k}}\subset (\F^m)^k$.  We  identify $\F^{mk}$ with  $Mat_{mk} (\F)$. Formula (\ref{eq:ev1})
says that for $i \in [1,m],\,j \in [1, k]$ the $(ij)$-component of  the image of the  evaluation mapping $Ev^r_{s,m,k}$   equals $ \tilde f_i^*( a_j)$. 
Using  the monomial basis  for $Pol^r_{hom} (\F^s)$ we represent $ \tilde f_i$  in coordinates  as $(\tilde f_i ^{\alpha})$, $\alpha\in [1, \binom {r+s -1}{r}]$.
We also represent $\alpha$ as a multi-index  $\alpha = \alpha_1 \cdots \alpha_s$  where $\sum_{q =1}^s \alpha_q = r$.
We  write $a_j = (a_j ^ p)$, where  $p \in [1, s]$  and $a_j ^p \in \F$. Then
\begin{equation}
 (\tilde f_i^{\alpha_1\cdots \alpha_s})^*  (a_j) = \sum_\alpha (\tilde f_i^ {\alpha_1 \cdots \alpha_s})^* [(a_j ^1)^{\alpha_1} \cdots (a_j ^s)^{\alpha_s}].\label{eq:ev2}
\end{equation} 
(For each $j$ the coordinates $a_j  ^1, \cdots,   a_j ^s$  form   a basis of $(\F^s) ^* = Pol ^1_{hom}(\F^s)$.  Thus      the monomials $\{ (a_j ^1)^{\alpha_1} \cdots (a_j ^s)^{\alpha_s}| \, \sum_{q =1}^s \alpha_q = r\}$   form  a basis of $Pol^r_{hom} (\F^s)$.)

For  $r =1$ we have $Pol^1_{hom} (\F^s)^* = \F^s$ and the   evaluation mapping is a quadratic map. Furthermore, the above representation of $Ev^k_{s,1,m}$ is the usual matrix multiplication $Mat _{ms}(\F) \times Mat_{sk}(\F) \to Mat_{mk}(\F)$. 

The following Proposition  is well-known; its  proof is based on the fact that    the rank of  a matrix is equal to the rank of the span of its  column vectors and equal to the rank of the span of  its line  vectors.

\begin{proposition}\label{prop:rank}  If $s\le \min (k, m)$ then 
the image  of $Ev ^1_{s,  m,k}$ consists  of exactly  of matrices  of rank at most $s$  in $Mat_{mk} (\F)$.  If $s \ge \min(k,m)$  then
  $Ev^1_{s,m,k}$ is surjective.
\end{proposition}


Proposition \ref{prop:rank} tells us that the  image of $Ev ^1_{s,m,k}$ is a determinantal variety if $s\le \min (k, m)$. The generators of the ideal $I(Ev ^1_{s,m,k}(Mat_m ^s  \times Mat ^s_k))$
are minors of rank $(s+1)\times (s+1)$.


Now let us consider the  case $ r\ge 1$.  Note that
\begin{equation}
Ev ^r_{s,m,k}  ( \tilde f^*, a) = Ev^1_{s,  m,k}\circ (Id, \nu_r ^k) (\tilde f^*, a), \label{eq:iter1}
\end{equation}
where $\tilde f^* \in (Pol^r_{hom}(\F^s)^*)^m$, $a \in (\F^s)^k$  and  $\nu^k_r$ is the  direct sum  of $k$ copies  of the Veronese map $\nu_r$,
$$\nu_r^k: (\F^{s})^k \to (F^{b(s+r-1,s)})^k, (a_1, \cdots, a_k) \mapsto (\nu_r(a_1), \cdots, \nu_r(a_k)). $$
Let us describe the ideal of the  image of  the polynomial mapping $(Id, \nu_r^k)$. 
It is known that (see e.g. \cite[p. 23]{Harris1993})
\begin{equation}
I(\nu_r(\F^s)) = \la  (x^{\alpha_1\cdots \alpha_s}x^{\beta_1 \cdots \beta_s} -  x^{\gamma_1\cdots \gamma_s} x^{\delta_1\cdots \delta_s})|\, X^{\alpha}X^{\beta} = X^{\gamma} X^{\delta}\ra_{\F[x^{\alpha_1 \cdots \alpha_s}]}
\end{equation}
where $X^\alpha$ denotes  the monomial $x_1^{\alpha_1} \cdots x_s ^{\alpha_s}$ corresponding  to  the multi-index $\alpha = \alpha_1\cdots \alpha_s$ and
$\{x^{\alpha_1\cdots \alpha_s}\}$ is a basis of $\F^{b(s+r-1,s)}$.
Next, we observe that
\begin{equation}
I(Id, \nu_r^k)((Pol^r_{hom}(\F^s)^*)^m, (\F^s)^k) = \la \oplus_{i =1} ^k I_i(\nu_r(\F ^ s))\ra. \nonumber
\end{equation}
We regard  elements of $(Pol^r_{hom}(\F^s)^*)^m$  as  matrices  over $\F$ of  size $Sm$, $S = \binom{s+r-1}{s}$,  and elements of
$(\F^S)^k$  as matrices   over $\F$ of size  $Sk$. Summarizing we have
\begin{eqnarray}
\ker Ev ^r_{s,m,k} = \{ \tilde  g \in Pol^*_{hom} (Mat _{mk} (\F))|\,  \tilde g([f^{*,l_1 \cdots l_s} _j] \cdot [ x^{l_1, \cdots l_s}_t]) \in \nonumber\\  
 I(Id, \nu_r^k)((Pol^r_{hom}(\F^s)^*)^m,(\F^s)^k)\}.\label{eq:ker}
\end{eqnarray}

The   identity  (\ref{eq:ker})   serves as  a starting point for  our future  work  on Problem 3.

\

In what follows we show the existence  of
many weakly elusive  homogeneous polynomial mappings.

\begin{theorem}\label{thm:exist}  Assume that $s\le m-1$ and
\begin{equation}
\binom{n+p-1}{p} \ge \frac{m\binom{s+r -1}{r} }{ m -s }.\label{eq:codi2}
\end{equation}
If $char(\F) = 0$ or $ p \le char (\F) -1$, then  the image of  almost  every  (i.e., except a    subset  of codimension at least 1) homogeneous polynomial  mapping $P \in Pol^p_{hom}(\F^n, \F^m)$
  contains  a  $k$-tuple  of points in $\F^m$ that is $(s,r)$-weakly elusive, where $k  = \binom{n+p-1}{p}$.
\end{theorem}

\begin{proof} Assume that $char (\F) = 0$ or $p \le char (\F) -1$. In \cite[Corollary 2.8]{Le2010} we provided a linear  isomorphism
$$I^p_{n,m}: Pol^p_{hom}(\F^n, \F^m) \to \F^{m\cdot b(n+p-1,p)},$$
using the interpolation formula \cite[Proposition 2.6]{Le2010}, which has the following property.  Let
$S_{n,p,m} \in \F^{m\cdot b(n+p-1,p)}$ be a tuple  of $\binom{n+p-1}{p}$  points  in $\F^m$.
Then  $(I^p_{n,m})^{-1}(S_{n,p,m})$ is a homogeneous polynomial mapping of degree $p$ from  $\F^n$ to $\F^m$ 
whose image is an algebraic  subset of $\F^m$ that   contains  all  the $\binom{n+p-1}{p}$ points   of the   tuple $S_{n,m,p}$.
Thus, to prove Theorem \ref{thm:exist}, it suffices  to show that  almost every (up to a subset of codimension at least 1)  tuple $S_{n,p,m}$ of $\binom{n+p-1}{p}$  points  in $\F^m$
is  $(s,r)$-weakly elusive, if (\ref{eq:codi2}) holds.

We note that 
\begin{equation}
\dim [(Pol_{hom} ^r (\F^s) ^*)^m \times  (\F^s) ^ k ] \le m\binom{s+r-1}{r}+ k\cdot s.\label{eq:codi3a}
\end{equation}
The equality in (\ref{eq:codi3a}) holds if $char(\F) = 0$  or $r \le char (\F) -1$. Since the  evaluation mapping $Ev ^r_{s,m,k}$ is bi-homogeneous of degree $(1, r)$, we derive  from (\ref{eq:codi3a}) that the image of the  evaluation map  $Ev_{s,r,m} ^ k$ is  a subset  of codimension at  least 1, if we have
\begin{equation}
m\binom{s+r-1}{r}+ k\cdot s  \le mk. \label{eq:codi3}
\end{equation}
Now assume that  (\ref{eq:codi2}) holds. Then for $k = \binom{n+p-1}{p}$,   the condition (\ref{eq:codi3}) holds.  Hence,  almost every (except  a subset of codimension at least 1)   point $\overline{S_{n,p,m}} \in \F^{m\cdot (n+p-1,p)}$  lies outside    the image  of  the evaluation map $Ev_{s,r,m}^k$, or equivalently,   by Proposition 
\ref{prop:keluh},  $S_{n,p,m}$ is an $(s,r)$-weakly elusive tuple  of points  in $\F^m$.
This completes the proof  of Theorem \ref{thm:exist}.
\end{proof}

\section{Applications}\label{sec:4}
In this section we introduce the notion of a    multivariate polynomial  that is {\it homogeneous  relative to a subset of its variables} (Definition \ref{def:ph}, Example \ref{exam:ph}).  
 We associate with each  polynomial $\tilde f$  that is homogeneous  relative to a subset of its variables a series of natural  polynomial families of  $m$-tuples  of homogeneous polynomials,  whose circuit  size   is bounded  from above  by the circuit size  of  $\tilde f$ (Proposition \ref{prop:natu}).
  As  a consequence,   we     estimate   from   below the circuit size of     $\tilde f$ in terms  of  the weak elusiveness of the associated polynomial mapping (Corollary \ref{cor:est1}).  For   a  large  class of  polynomials  $\tilde f$  
  our estimates  are non-trivial  (Examples  \ref{ex:raz}, \ref{ex:bi}.)  Using  Corollary \ref{cor:est1}, we  suggest a method  for obtaining  lower bounds  for the circuit size  of the permanent $P_n$  over  a  field  $\F$ of characteristic 0 (Lemmas \ref{lem:span}, \ref{lem:per}).
\subsection{Polynomial  families of $m$-tuples  of  homogeneous polynomials associated  with a  partially homogeneous   polynomial}
\begin{definition}\label{def:ph} A  multivariate polynomial $\tilde f \in \F[x_1, \cdots, x_n]$ will be called  {\it  homogeneous of degree  $d$ relative to  a non-empty proper subset $Z= \{  x_{k +1}, \cdots, x_{k+l}\} $ of  the set of variables $(x_1, \cdots, x_n)$}   if
$$ f(x_1, \cdots, x_k, \lambda  x_{k+1}, \cdots, \lambda x_{k+l}, x_{k+l+1}, \cdots , x_n) =  \lambda ^d f(x_1, \cdots, x_{k+1}, \cdots, x_{k+l}, \cdots, x_n)$$
for all $\lambda \in \F$. 
\end{definition}

For a   set  $Z$  of variables  let us denote  by $\F\la Z \ra $ the vector space  over $\F$ whose  coordinates  are   the variables in $Z$.

\begin{example}\label{exam:ph} 1.  For $i\in [1,n]$ let $Z =  Z_i : =\{ x_{ij}|\, j \in [1,n]\}$.  Then  the permanent  $Per_n$
is  homogeneous of degree $1$ relative to $Z$.

2. Let $\tilde f \in Pol^r_{hom} (\F \la Z \ra)$  and $\tilde g \in Pol^* (\F \la   Y \ra )$. Then $\tilde g \cdot \tilde f$ is  homogeneous of degree $r$ relative to $Z$.
  
3. The polynomial $\tilde f = x^2 + y^2$  is  not  homogeneous  relative  to $Z = \{ z\}$.
\end{example}

 
Now assume that $\tilde  f \in \F[x_1, \cdots, x_n]$  is  a    polynomial that is homogeneous of degree $r$ relative to  a  proper subset $Z$  of its variables. 
We  shall  associate  with $\tilde f$ a   series of polynomial families $P_\lambda (\tilde f)$ of $m$-tuples of homogeneous polynomials whose circuit size is controlled   from above by the circuit size of $\tilde f$.

 W.l.o.g. we assume  that the circuit size $L(\tilde f)$ of $\tilde f$ is larger than $\#(Z)$. 
Set
$$Z^\perp : = \{ x_1, \cdots, x_n\} \setminus  Z.$$
Since $Z$ is a proper  subset, $Z ^\perp$ is not  empty. Let $X$ be a subset  of $Z ^\perp$ such  that   for each $x_i \in  X$  the polynomial $P$ has exactly  degree 1 in $x_i$.  This set $X$ may be  empty  and need not to be
the subset  of all  variables   $x_j$ of degree  1 in $P$.

Let \begin{itemize}
\item $Y: = Z^\perp \setminus  X$;
\item  $p: = \#(X)$,    $k : = \#(Y)$  and $l : = \#(Z)$;
\item  $r$: = the  total degree of $\tilde f$ in $Z$; 
\item $m' := \dim Pol^r_{hom}(\F\la Z \ra)  = \binom{l +r -1}{r}$;
\item $m: = m'$ if $X$  is an empty set. If not, set $m: = p \cdot m'$;
\item  $h:[1,m'] \to Pol^r_{hom}(\F\la Z\ra)$  an  ordering of the monomial basis.
\end{itemize}

Regarding  $X$  as a parameter, we  shall associate  with $\tilde f$  a polynomial   family of ($p$-tuples)  of  homogeneous  polynomials in variables $Z$  by specifying 
the associated polynomial  mapping $f$  as  follows.
\

{\it Case 1.}  Assume that $X$ is an empty set, so $m = m'$. Then $\tilde f$ is a polynomial in   variables $Y, Z$. Now  we write $\tilde f$  as follows
\begin{equation}
\tilde f (x_1, \cdots, x_n): = \sum_{q =1} ^m \tilde f_{q} (Y) h(q),\nonumber 
\end{equation}
where $\tilde f_{q}  \in Pol^*(\F \la Y \ra )$. We associate  with $\tilde f$  the following polynomial mapping $f: \F ^k  \to Pol^r _{hom} (\F \la Z \ra )$:
\begin{equation}
f (\lambda) : =\sum_{q=1} ^m \tilde f_{q}(\lambda) h(q).\label{eq:pt1}
\end{equation}

\

{\it Case 2.} Assume that $X$ is not empty, i.e. $p \ge 1$.  Let us enumerate  the polynomials in the  set
$\{{ \p \tilde f \over \p x}, x\in X\}$  by $\tilde f_1, \cdots, \tilde f_p$.  For $j \in [1, p]$, we write $\tilde f_j\in  Pol^*(\F\la Y, Z\ra)$ as  follows
$$\tilde f_j( Y, Z) : = \sum_{q = 1} ^{m'} \tilde f_{j, q} (Y) h (q),$$
where $\tilde f_{j, q} \in Pol^*(\F \la Y \ra )$. We associate with $\tilde f$ the following polynomial mapping $f: \F ^k  \to (Pol^r _{hom} (\F \la Z \ra ))^p$:
\begin{equation}
f (\lambda) : = (\sum_{q=1} ^{m'} \tilde f_{1,q}(\lambda) h(q),   \cdots , \sum_{q=1}^{m'}  \tilde f_{p, q}(\lambda) h(q)) \in (Pol^r_{hom} (\F\la Z \ra )) ^ p .\label{eq:pt2}
\end{equation}

\begin{example}\label{ex:per} Let $n$ be  a basis parameter.  We shall apply  the  above construction to the  permanent $Per_n$.
We fix an additional  parameter $2\le t\le n-2$.  Then   we partition the  set of variables $\{ x_{ij}, 1\le i, j\le n\}$ of   the permanent $Per_n$
into  three subsets   $X$, $Y$, $Z$  as follows
$$ X = \{ x_{1i}, \, i\in[1,n]\},$$
$$ Y = \{ x_{ui}, \, 2 \le u \le t,\,  i \in [1,n]\},$$
$$ Z = \{ x_{ui}, \, t+1 \le  u \le n, \,i \in [1,n] \}.$$
\begin{itemize}
\item Set $m': = \dim  Pol^{n-t}_{hom} (\F\la Z \ra ) = \binom {(n-t)(n +1)  -1 }{n-t}$. 
\item Set $m : = n \cdot m'$. 
\item Let $h: [1,m'] \to Pol^{n-t} _{hom} (\F\la Z\ra)$  be  an ordering of the monomial  basis.
\end{itemize}
Since $\#(X)=n \ge 1$,  we  are in the  Case 2.  We represent the permanent as  follows
\begin{equation}
Per_n ([x_{ij}]) = \sum_{i =1} ^n  x_{1i}P_{n-1, i} (Y, Z),\label{eq:sum1}
\end{equation}
where $P_{n-1, i}(Y, Z)= {\p Per_n \over \p x_{1i}}$.
For each  $i \in [1,n]$  there is a unique decomposition
$$P_{n-1, i} (Y, Z) = \sum_{q =1}^{m'} \tilde f_{n-1, i, q}  (Y) h (q), $$
where $\tilde f_{n-1, i, q} \in Pol ^{t-1}_{hom} (\F ( \la  Y \ra ))$.
Note that $\#(Y) = (t-1)n$.  We  now associate  with  the permanent $Per_n$  and with the partition  of the variables  of $Per_n$ the following polynomial     mapping 
\begin{equation}
f : \F\la Y\ra \to (Pol ^{n-t}_{hom} (\F\la Z \ra))^n\label{eq:f}
\end{equation}
by the above recipe. Its $i$-th component $f_{i}\in  Pol ^{n-t}_{hom} (\F\la Z \ra) $, for $ i \in [1, n]$, is defined as follows (cf. \ref{eq:pt2}):
\begin{equation}
 f_{i} (\lambda_{21}, \cdots, \lambda_{tn}) : = \sum_{q =1} ^{m'}\tilde f _{n-1, i,q} (\lambda_{21}, \cdots, \lambda_{tn})  h ^{-1}(q).\label{eq:aspm}
 \end{equation}
 
\

Now we  make another  partition of  the  set of variables $\{ x_{ij}, 1\le i, j\le n\}$ of   the permanent $Per_n$ into three subsets $X', Y', Z'$, where $X'$ is the empty;  in    other words, we are in the  Case 1. Let  
$$Y':= \{ x_{ui} | 1\le u \le t, \, i \in [1, n]\},$$ 
$$Z': = \{ x_{ui}| \, u+1 \le  i \le n, \,i \in [1,n] \}.$$
Note that $\#(Y') = tn$. We associate with the permanent $Per_n$  another   family  of polynomial mappings
$ f:  \F^{tn}  \to  Pol ^{n-t} _{hom} (\F \la Z' \ra )$, using the  recipe  in (\ref{eq:pt1}): 
\begin{equation}
 f(\lambda_{11}, \cdots, \lambda_{tn}) :=\sum_{q=1}^m \tilde f_{q} (\lambda_{11}, \cdots, \lambda_{tn}) h (q).\label{eq:gtilde}
\end{equation}
Here $\tilde  f_{q} \in Pol ^{n-t}_{hom} (\F \la Y' \ra)$ is defined uniquely from the  equation
$$Per_n (Y', Z') = \sum_{ j = 1} ^{m '} \tilde  f_{n,j} (Y')   h (j).$$
\end{example} 



\

The following  Proposition shows that the  circuit size  of each member in the  polynomial  families   of tuples of homogeneous polynomials of  equal  degree that is associated with a  polynomial $f$  which is homogeneous relative   to a subset of its variables  is bounded  by the circuit size  of $P$.

\begin{proposition}\label{prop:natu}  1. Let $f : \F^k  \to Pol ^r _{hom}(\F\la Z \ra )$ be the polynomial mapping  in (\ref{eq:pt1}).
Then for  each $\lambda \in \F^k $  the circuit size  of  the  polynomial $f(\lambda)$ is at most  $L(\tilde f)$.

2.  Let $f: \F ^k  \to  (Pol^r_{hom}(\F \la Z \ra )^p$ be the polynomial mapping  in (\ref{eq:pt2}).  Then 
for  each $\lambda \in \F^k $  the circuit size  of  the $p$-tuple $f(\lambda)$ of homogeneous polynomials of degree $r$  is at most  $5L(\tilde f)$.
\end{proposition}

\begin{proof} 1. Let us consider  the case that $f$ is defined  by (\ref{eq:pt1}).
Note that for each $\lambda \in \F ^k $, we have
$$f(\lambda)(Z) = \tilde f(\lambda, Z) \in Pol^r _{hom} (\F\la Z \ra ).$$
It follows that the  circuit size $L(f(\lambda))$  is  at most
$L(\tilde f)$, what is required to prove. 

2. By  the Baur-Strassen  result \cite{BS1983}, there  exists  an arithmetic  circuit  $\Phi$ of   size  less than $5L(\tilde f)$  that computes the $p$-tuple 
$$\{ {\p \tilde f\over \p x} |\, x \in X\} = \{\tilde f_i(Y, Z) \in Pol ^{*} (\F^{*}(\la Y , Z\ra )) |\, i = [1, p]\} .$$
Note that for any value $ \lambda\in \F^k$    we have
$$ f(\lambda) = (\tilde f_1 (\lambda, Z),\cdots, \tilde f_p(\lambda, Z) ),$$
 which is  an $p$-tuple of polynomials in $Z$ of circuit size less than  or  equal to  $Size (\Phi)$. Since  $Size (\Phi)< 5L(\tilde f)$, we obtain 
 $$L(f(\lambda)) < 5L(\tilde f) \text{ for all } \lambda\in \F^k.$$
 This completes  the proof  of  Proposition \ref{prop:natu}.
\end{proof}

Combining Proposition \ref{prop:natu} with Corollary \ref{cor:his}, we obtain immediately

\begin{corollary}\label{cor:est1} 
1. Assume that  the   polynomial mapping $f$ defined  by the recipe in (\ref{eq:pt1})   is $(s,2r-1)$-weakly elusive. 
Then  the circuit size $L(\tilde f)$ of the associated  polynomial $\tilde f$ satisfies
$$L(\tilde f)  >  {\sqrt s\over 16 \cdot  (r+2)^{3}}.$$
2. Assume that  the   polynomial mapping $f$ defined in (\ref{eq:pt2})  is $(s,2r-1)$-weakly elusive. 
Then  the circuit size $L(\tilde f)$  of the  associated  polynomial $\tilde f$ satisfies
$$L(\tilde f)  >  {\sqrt s\over 80\cdot (r+2)^{3}}.$$
\end{corollary}

\begin{example}\label{ex:raz} Let us consider  an example from \cite[\S 3.3, \S 3.4]{Raz2009}, which motivates our  construction in (\ref{eq:pt2}).
Let $m': = \binom{n +r -1}{r}$  and $m = n \cdot m'$.
Assume that we are  given $f_{q, i} \in \F[x_1, \cdots, x_n]$, where $ q \in [1, m']$ and  $i \in [1,n]$. 
As before, let $h : [1,m'] \to  Pol^r_{hom} (\F^n)$  be  an ordering of the monomial basis  of  $Pol^r_{hom} (\F\la z_1, \cdots, z_n\ra)$. 
For $i \in [1,n]$ we define  $\tilde f_i \in \F[x_1, \cdots, x_n, z_1, \cdots , z_n]$ by
\begin{equation}
\tilde  f_i (x_1, \cdots, x_n, z_1, \cdots z_n) : = \sum_{q=1} ^ {m'} f_{q, i} (x_1, \cdots , x_n) h(q).\label{eq:tfraz}
\end{equation}
Let $W: = \{ w_1, \cdots, w_n\}$  be an additional  set of variables.  We define $\tilde f \in \F[x_1, \cdots, x_n, z_1, \cdots, z_n, w_1, \cdots w_n]$ as follows
\begin{equation}
\tilde f (x_1, \cdots, x_n \cdots, y_1, \cdots y_n, z_1, \cdots , z_n): = \sum_{i=1}^n  w_i \tilde f_i (x_1, \cdots  x_n, z_1, \cdots, z_n).\label{eq:fraz}
\end{equation}
Note that  $\tilde f$ is  homogeneous   of degree $r$ relative to  the proper  subset $Z := \{ z_1, \cdots, z_n\}$. Next we note that $\tilde f_i = {\p \tilde f\over \p w_i}$.
Now we  construct  $f: \F^n \to \F^m$   according to the recipe (\ref{eq:pt2}): 
$$f(\lambda): = (\sum_{q=1}^{m'} f_{q, 1} (\lambda), \cdots , \sum_{q+1}^{m'} f_{q, n} (\lambda)).$$
Corollary \ref{cor:est1}.2  implies immediately

\begin{lemma}\label{lem:cor38} (cf \cite [Corollary 3.8]{Raz2009}) Let $1\le r \le n \le s$ and $m  = n \cdot \binom{n+r-1}{r}$ be integers.  Let $f: \F^n \to \F^m$ be a polynomial mapping. If $f$ is $(s, 2r-1)$-weakly elusive, then   the circuit size of the polynomial $\tilde f$ defined by (\ref{eq:fraz}) and (\ref{eq:tfraz}) is  at least ${\sqrt s\over 80\cdot (r+2)^{3}}$.
\end{lemma}
\end{example}

\begin{remark}\label{rem:cor38} Lemma \ref{lem:cor38}  is an improvement  of  \cite[Corollary 3.8]{Raz2009}, which, under the assumption of  Lemma \ref{lem:cor38},  provides  the lower bound  $\Om (\sqrt s/ r^ 4)$ for the circuit size  of $\tilde f$. This lower bound is  weaker than our  lower bound (${\sqrt s\over 80\cdot (r+2)^{3}}$). Our improvement is due to  our upper bound $256 \cdot s^2 (r+2)^6$ for    the number of
of edges leading to the sum-gates   in   the universal circuit-graph $G_{s,r,n,m}$, see  Proposition \ref{prop:universal},   which is better than  the upper bound $\Om(s^2\cdot r^8)$  obtained by  Raz in \cite[Proposition 3.3]{Raz2009} 
for the number of the nodes in the universal graph-circuit $G_{s,r,n,n}$. 
\end{remark}

\begin{corollary}\label{cor:412} Let $char (\F) = 0$. Assume that $r$ grows   much slower  than $n$, e.g. $ r = const$ or $ r =  \ln \ln n$. Let $p = (r -1) (2r-1)$.  Then there are sequences  of  polynomials $\tilde f_n \in Pol^{p + r +1}_{hom} (\F^{3n})$ whose coefficients  are algebraic numbers, such that
$$L(\tilde f_n) \ge  (\frac{\lfloor {n \over  r(r -1)} \rfloor  ^{ {r-3\over 2}}}{ 80(r+2)^3}).$$
\end{corollary}

Corollary  \ref{cor:412}  and its  proof  are almost identical  with Corollary 4.12 and its proof in \cite{Le2010}, except that  the estimate in Corollary \ref{cor:412}  is better than the one in  Corollary 4.12 in \cite{Le2010} (the dominator  contains  $(r+2)^3$ vs  $r^4$). This improvement  is  due to  Lemma \ref{lem:cor38}, which replaces    Corollary 3.8 in \cite{Raz2009}  in the proof  of Corollary 4.12 in \cite{Le2010}.  We refer the reader to \cite{Le2010} for  the proof  of \cite[Corollary 4.12]{Le2010}  and omit the proof  of  Corollary \ref{cor:412}.

\begin{example}\label{ex:bi} Let  $X$, $Y$ be a set  of variables  and $n_X = \#(X)$, $n_Y = \#(Y)$.  We  denote  by $Pol^{p, q}_{hom}(\F \la X, Y\ra)$ the  linear  subspace of 
$Pol^{p+q}_{hom}(\F\la X, Y \ra)$  consisting  of all  polynomials which are homogeneous of degree $p$ in   $X$  and homogeneous of degree $q$  in $Y$.  For example, the permanent  $P_n$ belongs to  $Pol ^{t, n-t}_{hom} (Y', Z')$,
where $Y', Z'$ are defined  in Example \ref{ex:per}.  Then  each  polynomial $\tilde f \in Pol^{p,q}_{hom} (X, Y)$   is  associated uniquely by the recipe  of (\ref{eq:pt1})
with a  homogeneous   polynomial mapping $f \in Pol ^p_{hom} (\F\la X \ra, Pol^{q}_{hom} (\F \la  Y\ra))$.  Now assume that $s +1 \le \binom{n_Y + q -1}{q}$ and
\begin{equation}
\binom{n_X + p -1}{p}\ge \frac{\binom{n_Y + q -1}{q} \binom{s + 2q -2}{2q-1} }{\binom{n_Y + q -1}{q} - s }.\label{eq:bih}
\end{equation}
Theorem \ref{thm:exist} implies that, if $char (\F) = 0$ or $p \le char (\F) -1$, almost all   polynomial $\tilde f$  in $Pol ^p_{hom} (\F\la X \ra, Pol^{q}_{hom} (\F \la  Y\ra))$ is $(s, 2q-1)$-elusive,  and hence, by Corollary \ref{cor:est1}.1 we have
\begin{equation}
L(\tilde f) \ge \frac{\sqrt s}{16(q+2) ^3}.\label{eq:bih2}
\end{equation}
Furthermore, assume that $char (\F) = 0$.  Then using  Proposition 4.5 in \cite{Le2010}, adapted to our case  of  weakly  elusive  homogeneous  polynomial mappings, 
it is easy to  exhibit  explicitly  infinitely many   bi-homogeneous    polynomials $\tilde f \in Pol ^{p, q}_{hom} (\F\la X ,  Y\ra)$ whose  monomials coefficients are algebraic numbers   such that $\tilde f$ satisfies  (\ref{eq:bih2}), if  (\ref{eq:bih}) holds.  We refer  the reader to  the proof  of Proposition  4.5  in \cite{Le2010}  for the method of the proof of this assertion.
\end{example}

\subsection{An approach   for obtaining  lower bounds  for the circuit size of the  permanent}\label{sub:per}
In this subsection we suggest a method for  obtaining   lower bounds of $L(Per_n)$  using the  mappings $f: \F \la Y \ra Y \to Pol ^{n-t}_{hom} (\F \la  Z\ra ) $  defined in (\ref{eq:gtilde})  that  are associated with $Per_n$. 
We  keep the notations in Example \ref{ex:per}  and  assume that $\F$ is a field of characteristic 0.
Recall that $Z$   is a rectangular  matrix of size  $(n-t)\times n$.   Denote by $\bar Per_{n,t}(Z)$ the  linear subspace  of $Pol ^{n-t} _{hom}(\F\la Z \ra )$  which  is generated by
the (minor) permanents  of size $(n-t) \times (n-t)$  of the  matrix $Z$.   Clearly  $\dim \bar Per_{n,t} (Z) = \binom{n}{n-t}$. 
Let us denote by $Per_{n,t} (Y, Z)$ the  linear subspace  in $Pol ^{t, n-t}_{hom} (\F \la Y, Z\ra)$ consisting   of all polynomials $P$  whose associated polynomial mapping $\tilde P$ takes values  in $\bar Per_{n,t} (Z)$.  The following Lemma \ref{lem:red} infers that  to study the  weak elusiveness of the  polynomial mapping $f$  we need to know the smallest linear   subspace in  $Pol ^{n-t}_{hom} (\F \la  Z\ra ) $ that contains the image  of $f$.

\begin{lemma}\label{lem:red}  Assume that $f: \F^n \to \F^m$ is $(s,1)$-weakly elusive and not $(s+1, 1)$-weakly elusive, i.e.  the linear  span of the image
of $f(\F^n)$  is a linear  subspace $\F^{s+1}$ in $\F^m$.  Let $\pi: \F^m \to \F^{s+1}$ be a  projection.
Then $f$ is $(l,r)$-weakly elusive, if and only if $\pi \circ f: \F^n \to \F^{s+1}$  is $(l, r)$-weakly elusive.
\end{lemma}

\begin{proof} First  let us prove the ``only if" assertion. Assume that $f$ is $(l,r)$- weakly elusive  and $\pi \circ f$ is not $(l,r)$-elusive.
Then there  exists  a homogeneous polynomial mapping $\Gamma: \F^l \to \F^{s+1}$ of degree $r$ such that  $\pi \circ f (\F^n)$ lies on the image
of $\Gamma (\F^l)$. Let $i : \F^{s+1} \to  \F^m$ be  the  embedding, such that $\pi\circ i = Id$. It  follows that $f(\F^n)$  lies on the image
of the  map $i \circ \Gamma (\F^l)$ of  degree $r$,  which contradicts    our assumption. This proves  the ``only if"  assertion.

Now let us prove  the ``if" assertion.  Assume that $f$  is not $(l, s)$-weakly elusive, i.e.  the image $f(\F^n)$ belongs to  the image $\Gamma (\F^l)$  for some homogeneous  polynomial mapping  $\Gamma: \F^l \to \F^m$ of degree $r$. Then the image $\pi \circ f (\F^m)$  belongs to the image of
$\pi\circ  \Gamma : \F^l \to \F^{s+1}$. Hence  $\pi \circ f$ is not $(l, s)$-elusive. This completes the proof of Lemma \ref{lem:red}.
\end{proof}

The following   Lemma   exhibits the smallest linear   subspace in  $Pol ^{n-t}_{hom} (\F \la  Z\ra ) $ that contains the image  of $f$.
\begin{lemma}\label{lem:span}   The linear  span of $f (\F(\la Y \ra )$  is  $\bar Per_{n,t} (Z)$. Hence $f$ is $(s, 1)$-weakly  elusive  for $s = \binom{n}{n-t} -1$.
\end{lemma}

\begin{proof} Note that the linear  span of $f(\F(\la Y \ra )$  belongs to $\bar Per_{n,t} (Z)$. Now we complete  the proof of  Lemma \ref{lem:span}  by observing that  all the  basis  of $\bar Per(Z)$  lies on the image of $f$.
\end{proof}

Let  $n,t,s$   be defined as follows   
\begin{equation}
 n- t = N,  \, t =  N^3 (N -1), \, n = N^ 4, \, k = constant , \, s = n ^k = N^{4k}.\label{eq:nts}
 \end{equation}
To   obtain a lower  bound  for the circuit  size of $Per_n$   we    might      prove the weak   elusiveness  of  $f$  associated with $Per_n\in Per_{n,t}(Y,Z)$. The following Lemma    says that  the polynomial  mappings associated  with ``generic"   
polynomials   in $Per_{n, t}(Y, Z)$  are  $(s, n-t)$-weak elusive, and hence  the ``generic"  polynomials  have  very  large circuit  size.
\begin{lemma}\label{lem:per} Given  any $k$    almost all  homogeneous  polynomials  in\\
 $Per _{n, t}(Y, Z)$  has the circuit size     at least
$$ \frac{ \sqrt s}{ 16\,  (n-t+2) ^3 }$$
if $N$ is sufficient large, since  their  associated    polynomial mappings  are $(s, n-t)$ elusive. 
\end{lemma}
\begin{proof} For  sufficiently large  $N\in \N$,  using the Stirling approximation $N! \sim \sqrt{2\pi N}({N\over e})^N$ 
we  obtain 
\begin{equation}
\binom{N^7}{N^3} \ge 2 \binom{N^{4k} + 2N}{2N}.\label{eq:nbig}
\end{equation}
 Let $(n,t,s)$ are given   as in (\ref{eq:nts}). Then (\ref{eq:nbig}) implies  the  following  inequality
  \begin{equation}
 \binom{nt + t -1}{t} \ge \frac{\binom{n}{n-t}\binom{s + 2 (n -t) -2}{2 (n-t)} }{\binom{n}{n-t} - s}, \label{eq:per}
 \end{equation}
 Theorem  \ref{thm:exist} implies  that, the validity of (\ref{eq:per}) implies the density of
  homogeneous   polynomials in $Per_{n,t}(Y, Z)$   that have the circuit size  at least  $\sqrt s/(16\, (n -t+2) ^ 3)$  since   their  associated  polynomial mappings 
  are $(s, n-t)$-elusive.  This  proves Lemma
  \ref{lem:per}.
\end{proof}

\section*{Acknowledgements}  The author  would like to thank  Pavel Hrube\v s, Partha Mukhopadhyay   and Pavel Pudlak for their stimulating discussions, helpful remarks  and suggestions,  and Ng\^o Vi\^et Trung   for his  illuminating
discussion on related  problems  in commutative  algebra.   She  is grateful to the anonymous referee  for valuable  suggestions.
A part of this paper has been conceived during the author's visit to VNU for Sciences  and the Institute  of  Mathematics   of  VAST  in Hanoi. She  would like  to thank these institutions for excellent working conditions  and financial support.

\end{document}